\documentclass[12pt]{article}

\usepackage{epsfig}

\usepackage{amssymb}
\usepackage{amsfonts}

\usepackage{color}

\def\be{\begin{equation}}
\def\ee{\end{equation}}
\def\ba{\begin{array}{c}}
\def\ea{\end{array}}

\def\ben{$$}
\def\een{$$}

\newcommand{\bea}{\begin{eqnarray}}
\newcommand{\eea}{\end{eqnarray}}

\newtheorem{thm}{Theorem}

\newtheorem{lemma}[thm]{Lemma}

\newenvironment{proof}{\noindent {\bf Proof}}{\hfill$\square$\vspace{3mm}\endtrivlist}

\newcommand{\bba}[1]{\begin{eqnarray}}
\newcommand{\eba}{\end{eqnarray}}

\begin{document}



 \begin{center}{\Large \bf

Resonances and continued-fraction Green's functions
in non-Hermitian
Bose-Hubbard-like quantum models

  }\end{center}

 \begin{center}

\vspace{8mm}

  {\bf Miloslav Znojil} $^{1,2}$

\end{center}

\vspace{8mm}

  $^{1}$
 {The Czech Academy of Sciences,
 Nuclear Physics Institute,
 Hlavn\'{\i} 130,
250 68 \v{R}e\v{z}, Czech Republic, {e-mail: znojil@ujf.cas.cz}}


 $^{2}$
 {Department of Physics, Faculty of
Science, University of Hradec Kr\'{a}lov\'{e}, Rokitansk\'{e}ho 62,
50003 Hradec Kr\'{a}lov\'{e},
 Czech Republic}


\section*{Abstract}

With resonances treated as eigenstates of
a non-Hermitian quantum Hamiltonian $H$,
the task of
localization of the complex energy eigenvalues of $H$
is considered.
The paper
is devoted to its
reduced version in which one only computes
the real quantities
called singular values
of $H$.
It is shown that in such an approach
(and under suitable constraints
including the tridiagonality of $H$)
the singular values can be
sought as poles of a ``Hermitized'' Green's function
expressible in terms of a doublet of matrix continued fractions.
A family of
multi-bosonic
Bose-Hubbard-like complex Hamiltonians
is recalled for illustration purposes.

\subsection*{Keywords}.

non-Hermitian complex tridiagonal Hamiltonians;

resonances in generalized Bose-Hubbard quantum systems;

singular values specified via an auxiliary
Schr\"{o}dinger-like equation;

analytic and matrix continued-fraction Green's functions;

\newpage

\section{Introduction}

Non-Hermitian quantum Hamiltonians $H \neq H^\dagger$
can generate the real or complex energy
spectra of bound or resonant
states $\{E_n\}$ \cite{Geyer,BB}. One of
the most remarkable many-body examples
of such a family of models
has been proposed and studied
by Graefe et al \cite{Uwe}. They
started their considerations from one of the standard
(viz., self-adjoint)  versions
of
the
three-parametric many-body
Bose-Hubbard Hamiltonian
\begin{equation} \label{Ham1}
 H^{(BH)}(\varepsilon,v,c) = \varepsilon\left(a_1^{\dagger}a_1
 - a_2^{\dagger}a_2\right) +
  v\left(a_1^{\dagger}a_2 + a_2^{\dagger}a_1\right) + \frac{c}{2}
  \left( a_1^{\dagger}a_1 - a_2^{\dagger}a_2\right)^2
\end{equation}
written in terms of the annihilation and
creation operators
$a$ and $a^\dagger$. They replaced
it by its analytically continued non-self-adjoint
avatar $H^{(BH)}({\rm i}\gamma,v,c)$
for which the spectrum can remain real
and compatible with the unitarity in a certain
subdomain of parameters, for
$(\gamma,v,c) \in {\cal D}^{(0)}$.

The
\textcolor{black}{above-cited}
proponents of their analytically continued model pointed out that
\textcolor{black}{during the
related experiment-oriented discussions
one has to distinguish, carefully, between
the phenomenological role
and observability aspects of the real and complex energies
or other physical observables. In the former case, indeed,
the bound-state interpretation is standard. In contrast,
an appropriate consistent interpretation of the complex energies
or other physical observables
remains model-dependent and
must be formulated case-by-case.}

\textcolor{black}{Typically,
the most common notion of a bound state
is being replaced by the concept of a resonance
(cf., e.g., \cite{Nimrod}).}
Out of the unitarity-supporting domain (i.e.,
\textcolor{black}{
in model (\ref{Ham1}),}
for $(\gamma,v,c) \notin \overline{{\cal D}^{(0)}}$
where the bar denotes the closure of the domain),
a consistent
physical interpretation of the system
\textcolor{black}{
is found in its open-system reinterpretation.}
In place of a unitary, closed quantum system
(called, sometimes, quasi-Hermitian \cite{Geyer}),
one has to
start speaking about an open
system
\textcolor{black}{defined as}
admitting the
loss of unitarity, say, due to the
influence of environment \cite{Nimrod}.

The unstable
eigenstates of any Hamiltonian with $E_n \notin \mathbb{R}$
have to be
reinterpreted as resonances, with their description
becoming a challenge.
The constructive
localization of the
complex energies
ceases to be easy.
In what follows we will address a
simplified task, therefore.
In a way inspired
by several recent studies (cf., e.g., \cite{PS1,MCF}) we
will turn attention to
the mere evaluation
of the so called singular
values $\sigma_n$ of $H$ which are real \cite{SV}.

In particular, we will
study the
possibility of a facilitated determination of singular values
via Green's functions.
For technical reasons, our attention will be further restricted just
to the Hamiltonians having a Jacobi-matrix form, finite or infinite.
Under this constraint,
the presentation of our results will start
in section \ref{osinghovi}.
For an overall preliminary illustration of applicability
of the specific tridiagonal Hamiltonians
in
quantum physics we
will recall there the most elementary discrete Schr\"{o}dinger equation
(a few related comments will be added in Appendix A below)
plus a slightly more sophisticated
anharmonic-oscillator example
of Ref.~\cite{Singh}.

We will remind the readers that
the popularity of the latter example
originated from its methodical relevance in
quantum field theory.
In a marginal addendum we will also point out that
the specific choice of an anharmonic-oscillator
example
helped to discover the existence of
the so called quasi-exactly solvable quantum
bound-state models \cite{Ushveridze}.

One of the main results
of Ref.~\cite{Singh} was that
for the particular
oscillator in question, the (real and discrete)
bound state energies $E_n$
of the system
could be defined as
poles of an {\it ad hoc\,} Green's function ${\cal G}(z)$.
This function was
obtained in the form of a convergent analytic continued fraction.
Here, such a construction will be used as a methodical guide.
We will briefly explain the
derivation of the formula, and we
will point out that a few similar
results will also form the main
contents of our present paper.

In section \ref{ufoun} we will emphasize that
our considerations were mainly inspired
by the phenomenological appeal of
the above-mentioned Bose-Hubbard
model $H^{(BH)}({\rm i}\gamma,v,c)$  (cf. also Appendix B below)
and of its
generalizations (cf., e.g., \cite{zobecUwe}).
First of all, we will show that the
specific formal features of all of these
Bose-Hubbard-like
models
lead to the necessity of an amendment of the
very method of the derivation
as well as of the resulting structure of the
corresponding analytic Green's function.

We will
reveal that such an amended Green's function
has to be defined in terms of {\em two\,}
continued fractions.
In this manner, some of
the limitations of
the single-continued-fraction
methods of Refs.~\cite{MCF} or \cite{Singh}
will successfully be removed.
In particular,
what will be explained and clarified
will be the reasons of
the failure of our recent repeated attempts
of
application of the single-continued-fraction
methods
to our present
broader family of
Bose-Hubbard-like systems comprising a larger number of particles~\cite{Gar}.

A formal climax of our present constructive
considerations will come in section~\ref{ultim}.
We will upgrade the idea of the
matrix-continued-fraction-based
reduction
of the
search for the energies of the resonances
to the search for
the mere singular values $\sigma_n$
of $H$. In a way inspired by Pushnitsky and \v{S}tampach \cite{PS1}
we will \textcolor{black}{introduce an auxiliary quasi-Hamiltonian
(with eigenvalues $\sigma_n$)}
and
we will block-tridiagonalize
\textcolor{black}{its matrix form}.
Ultimately, we will
end up with
the construction
of a related auxiliary Green's function
defined in terms of a pair of {\em matrix\,}
continued fractions.

A brief discussion will be then added in section \ref{cotos}
and the results will be reviewed in section \ref{summary}.

\section{Single-continued-fraction Green's functions\label{osinghovi}}

Methodical
appeal of
analytic representations of
resolvents and Green's functions ranges from
quantum chemistry to
quantum field theory.
A typical example
can be found in~\cite{Singh}.

\subsection{Example}


In {\it loc. cit.},
Singh,
Biswas and Datta
picked up
a particular (viz., sextic) quantum anharmonic oscillator,
and
they managed to
reformulate the related ordinary differential Schr\"{o}dinger equation
as a tridiagonal-matrix eigenvalue
problem
 \be
  \left( \begin {array}{ccccc}
     a_1&b_1&0
 &\ldots&0
   \\
   c_2&a_2&b_2&\ddots
 &\vdots
   \\
 0
 &\ddots&\ddots&\ddots&0
   \\
 \vdots&\ddots&c_{N-1}&a_{N-1}&b_{N-1}
    \\
  0&\ldots&0&c_{N}&a_{N}
    \\
 \end {array} \right)\, \left( \begin {array}{c}
 \psi_1\\
 \psi_2\\
 \vdots\\
 \psi_N\\
 \end {array} \right)
 =
 E
 \,\left( \begin {array}{c}
 \psi_1\\
 \psi_2\\
 \vdots\\
 \psi_N\\
 \end {array} \right)\,,
 \ \ \ \ N \to \infty\,
 \label{SEfinkit}
 \ee
characterized
by an anomalous
asymptotic asymmetry of
matrix elements,
 \be
 a_n={\cal O}(n)
 \,,\ \ \
 b_{n}=4n^2+{\cal O}(n)\,,
 \ \ \
 c_{n+1}={\cal O}(n)\,,\ \ \ \
 n \gg 1
 \,.
  \label{largeen}
  \ee
It was precisely
such a
feature of the model
which opened the way towards a transition
from the infinite algebraic set (\ref{SEfinkit})
to an
analytic
continued-fraction formula
 \be
 {\mathcal G}(z)=\frac{1}{
 a_1-z-\frac{b_1c_2}{a_2-z-\frac{b_2c_3}{a_3-z-\ldots}} }\,
 \label{Virendra}
 \ee
which defined its
Green's function.
A key to the success
lied in the rigorous
demonstration of the convergence of formula (\ref{Virendra}).
This, naturally, contributed also to the growth of popularity
of the analytic continued-fraction expansions
(cf., e.g., \cite{FP}).

Singh with coauthors recommended
formula (\ref{Virendra})
as an efficient means of
resummation
of
divergent Rayleigh-Schr\"{o}dinger
perturbation series.
The reduction of
problem~(\ref{SEfinkit})
to its
analytic Green's function (\ref{Virendra})
resulted from a lucky coincidence of several fortunate
circumstances but
some aspects of the
implementation of the recipe
were still intuitive, not
sufficiently rigorous and not always
leading to
the correct result
(cf., e.g., a critical comment in \cite{Hill}).

Even after the final resolution of the latter
problem by Hautot \cite{Hautot},
a few other
unpleasant obstacles
were encountered, e.g.,
when we tried to apply the idea
to
a fairly realistic
Bose-Hubbard
many-particle Hamiltonian  of paper \cite{Uwe}
(cf. Appendix A).
This motivated also our present paper.
We revealed that the
standard version of the continued fraction expansion of the resolvent
must be modified,
and we succeeded in finding its appropriate amendment.

The details will be described in what follows.
We will show that
in contrast to our initial expectations \cite{Hill},
the amendment of the approach
has not been related
to the rather anomalous
asymptotic asymmetry (\ref{largeen})
of the benchmark model of paper \cite{Singh}.

The occurrence of asymmetry has also been found inessential
in applications where
such a feature of Hamiltonian $H$ is only rarely encountered.
Most often,
the tridiagonality and symmetry
of the Hamiltonian matrix in
Schr\"{o}dinger Eq.~(\ref{SEfinkit})
is achieved when working
in the so called Lanczos' recurrently defined basis
in Hilbert space \cite{Lanczos,Nex}.


\subsection{Factorization method}


Let us now
consider Schr\"{o}dinger equations
of the form~(\ref{SEfinkit})
for which
the construction
of the continued-fraction Green's function (\ref{Virendra})
proved successful.
In the vast majority of models of such a type
the Hamiltonian is a matrix
which is well behaved in the limit $N \to \infty$ \cite{Haydock,GM}.
Then,
the convergence of the continued fraction (\ref{Virendra})
emerges as a consequence.

An explicit explanation and proof of the convergence
may be based on the
factorization
 \be
 H-E
 = {\cal U}\,
 {\cal F}\,
 {\cal L}\,
 \label{finkit}
 \ee
of Schr\"{o}dinger operator in which
the middle factor
${\cal F}$ is a diagonal matrix
with elements
 $$
 1/f_1,  1/f_2, \ldots,  1/f_N\,
 $$
%
%
%
and
in which
we require that
 \be
 \det  {\cal U}=
 \det  {\cal L}=1\,.
 \label{sestka}
 \ee
The two outer
factors may be chosen
two-diagonal,
 \be
 {\cal U}=
  \left[ \begin {array}{ccccc}
  1&-u_2&0
 &\ldots&0
   \\
     0&1&-u_3&\ddots
 &\vdots
   \\
 0
 &0&\ddots&\ddots&0
   \\
 \vdots&\ddots&\ddots&1&-u_N
    \\
  0&\ldots&0&0&1
    \\
 \end {array} \right]\,,
 \ \ \ \ \ \
 {\cal L}=
  \left[ \begin {array}{ccccc}
  1&0&0
 &\ldots&0
   \\
     -v_2&1&0
 &\ldots&0
   \\
   0&-v_3&\ddots&\ddots
 &\vdots
   \\
 \vdots&\ddots
 &\ddots&1&0
   \\
  0&\ldots&0&-v_{N}&1
    \\
 \end {array} \right]\,,
 \label{lowkit}
 \ee
After we specify
${u}_{k+1}=-b_kf_{k+1}$ and ${v}_j=-f_jc_j$,
this ansatz will simplify the
construction because it may convert Eq.~(\ref{finkit})
in identity, provided only
that we
evaluate the matrix elements of ${\cal F}$ via
recurrences
 \be
 f_k=\frac{1}{a_k-E-b_kf_{k+1}c_{k+1}}\,,\ \ \ \
 k=N, N-1,\ldots,2 ,1 
 \label{cf}
 \ee
with initial $f_{N+1}=0$.

In some models
(i.e., up to exceptions discussed in paragraph \ref{qesy} below)
it makes sense to
reconstruct also the
eigenstates. Schr\"{o}dinger equation
$$
(H-E_n)\,\overrightarrow{\psi^{}}=0
$$
is to be factorized,
 \be
 ({\cal U}\,
 {\cal F}\,
 {\cal L}\,)\,\overrightarrow{\psi^{}}=0\,.
 \label{1finkit}
 \ee
Both of the outer two factors are formally invertible,
and their
inversions can be
written in closed form,
 \ben
 {\cal U}^{-1}=
  \left[ \begin {array}{ccccc}
  1&{u}_2&{u}_2{u}_3
 &\ldots&{u}_2{u}_3\ldots{u}_N
   \\
     0&1&{u}_3&\ddots
 &\vdots
   \\
 0
 &0&\ddots&\ddots&{u}_{N-1}{u}_N
   \\
 \vdots&\ddots&\ddots&1&{u}_N
    \\
  0&\ldots&0&0&1
    \\
 \end {array} \right]\,,
 \ \ \
 {\cal L}^{-1}=
  \left[ \begin {array}{ccccc}
  1&0&0
 &\ldots&0
   \\
     {v}_2&1&0
 &\ldots&0
   \\
   {v}_3{v}_2&{v}_3&\ddots&\ddots
 &\vdots
   \\
 \vdots&\ddots
 &\ddots&1&0
   \\
   {v}_N \ldots {v}_3{v}_2&\ldots&{v}_N{v}_{N-1}&{v}_N&1
    \\
 \end {array} \right]\,.
 \een
In (\ref{1finkit}), the leftmost factor ${\cal U}$ can be omitted
as inessential, yielding Schr\"{o}dinger equation
 \be
 {\cal F}\,
 {\cal L}\,\overrightarrow{\psi^{}}=0\,.
 \label{2finkit}
 \ee
In a sufficiently small vicinity of the eigenvalue,
also all of
the elements of ${\cal F}$ are well defined
by recurrences~(\ref{cf})
so that the generic
secular equation $\det (H-E)=0$
degenerates to its continued-fraction form
 \be
 \frac{1}{f_1(E_n)}=0\,.
 \label{generi}
 \ee
At $E=E_n$, as a consequence,
the first line of the simplified Schr\"{o}dinger Eq.~(\ref{2finkit})
becomes satisfied identically.
Opening the freedom
of the choice of an arbitrary optional
normalization constant ${\psi^{}}_1\neq 0$
and, ultimately, yielding formula
 \be
 \overrightarrow{\psi^{}}=
 \left( \begin {array}{c}
 \psi_1\\
 \psi_2\\
 \psi_3\\
 \vdots\\
 \end {array} \right)={\cal L}^{-1}\,
 \left( \begin {array}{c}
 \psi_1\\
 0\\
 0\\
 \vdots\\
 \end {array} \right)\,
 \label{jedenact}
 \ee
which renders the wave function
defined in explicit recurrent manner.

 \textcolor{black}{In some sense, the simplicity of
the wave functions
as defined by recurrences (\ref{jedenact})
is an immediate consequence of the tridiagonality
(i.e., simplicity) of the underlying Hamiltonian
as well as of the (more or less tacitly postulated)
simplicity of the structure of the underlying
physical Hilbert space of states.
In Appendix A the readers can find an example in which
the specification of such a structure
ceased to be trivial. Another,
much less exotic
sample of such a specification
is, in multiple applications,
based on a certain {\it ad hoc\,} modification
of the asymptotic boundary conditions.
Interested readers may find
two characteristic samples of impact
of such modifications
in papers \cite{JCSa,JCS}:
In the former case
the quantum dynamics
is controlled by the
less usual but still standard
periodic boundary conditions
while in the latter paper the ``anomalous'' Robin
boundary conditions have been used.
In our present context, it is important to add that
after a minor modification,
the continued fraction approach remained applicable
in both cases.}

\subsection{Terminating recurrences and exact partial solutions\label{qesy}}

\textcolor{black}{In the practical numerical applications
of the Hermitian or non-Hermitian Hamiltonians
(but especially in the latter case)
one
encounters multiple numerical challenges when computing the poles of
the Green's function.
For the purpose, people often utilize the Krylov space tools
(cf., e.g., paper \cite{pratik}).
Obviously, such an  alternative treatment of
tridiagonality
could certainly yield
various new and interesting insights.
Indeed, some of them were formulated
in the latter reference where its authors
studied physical non-Hermitian systems
(including the so called Ginibre ensembles and
a non-Hermitian version of the Sachdev-Ye-Kitaev model)
which appeared to be tridiagonalizable through the
standard singular-value
decomposition \cite{SV}.}

\textcolor{black}{
In this sense, the paper contributed to an extension of
our
present project, with one of
the additional reasons being the fact that the  continued fractions
were used there for the purposes of evaluation of another relevant
quantity called Lanczos coefficient.
At the same time,
the fact that the paper
addressed several real-world phenomena
helped us to
acknowledge the limitations of their
analysis in the mathematical
continued fraction framework.
One must be aware of its limitations.
It makes also sense to add that all of the studies
which combine the pragmatic and theoretical aspects
are well
motivated by the enhancement of our understanding of
the dynamics and, in the particular
case of Ref.~\cite{pratik}, of certain integrable forms of dynamics
mimicking the manifestations of chaos
as generated by
non-Hermitian random matrices. }

\textcolor{black}{Incidentally, the underlying
formal feature of integrability may deserve one more
comment on
the contents of above-cited paper \cite{Singh}.
In this paper, indeed,}
the coefficients entering formula (\ref{Virendra}) were quickly growing
functions of the subscript
at large $n\gg 1$ (cf. Eq.~(\ref{largeen})).
From the point of view of
the analytic theory of continued
fractions \cite{Wall}
such a feature represented an important
merit of the model
which simplified the proof of the convergence
of the Green's function (\ref{Virendra}).
Nevertheless, in a special
dynamical regime
the preparatory
process of the
tridiagonalization
of the Hamiltonian in (\ref{SEfinkit})
led to the degeneracy
 \be
 c_{K+1}= 0\,
 \label{Roma}
 \ee
i.e., to
an accidental disappearance of
one of the coefficients
at a fixed finite $K <  N$.
This made
the continued-fraction expansion terminated
so that the secular Eq.~(\ref{generi}) became algebraic.
Also the related recurrences (\ref{jedenact})
for wave functions terminated
so that there was no need to study the limit $N\to \infty$.

The latter anomaly
(currently well known under a nickname of ``quasi-solvability'')
attracted
a lot of attention in the literature
(cf., e.g., an
extensive review \cite{Ushveridze}).
Paradoxically, its discovery in \cite{Singh}
which was just a byproduct of the possibility of
the termination (\ref{Roma}) of the
generic analytic continued-fraction formula (\ref{Virendra})
seems to be forgotten at present.

\section{Bose-Hubbard-like complex Hamiltonians\label{ufoun}}


In multiple realistic applications of quantum mechanics
the tridiagonal-matrix Hamiltonians need not necessarily
have the structure as sampled by
Eqs.~(\ref{SEfinkit}) + (\ref{largeen}) above.
In a way
explained in Appendix B and guided by
the influential paper \cite{Uwe}
we will turn attention,
in what follows,
to another class of models based, for the reasons explained below,
on the use of a broad family of
Bose-Hubbard-like Hamiltonians.

%
%

\subsection{Quantum models
admitting phase transitions}

Among multiple papers dealing with the
tridiagonal-matrix quantum Hamiltonians
(cf., e.g., their compact review in \cite{Marcelo})
we felt particularly addressed by
paper \cite{Uwe} in which the authors studied
the occurrence of quantum phase transitions mediated by the
singularities
called ``exceptional points'' (EP, \cite{Kato})
{\it alias\,}
``non-Hermitian degeneracies''
(cf. also \cite{Heiss} and \cite{Berry}).
Attention has been paid to
the quantitative description of the EP unfolding
tractable as a phase transitions (cf. also \cite{passage}).

The model used in \cite{Uwe} was a
complexified but still
realistic version of the
quantum Bose-Hubbard manybody system
(see more details in Appendix B).
In our older paper on such a subject \cite{zobecUwe}
we demonstrated that
one of the characteristic features
of the model
(viz., the existence of the EPs
of higher orders)
does not if fact require the existence of all of its
symmetries.
The generalized models have been shown to
cover both the closed-system (i.e., unitary)
and open-system (i.e., non-unitary) quantum dynamical regimes.

In detail we constructed several Bose-Hubbard-like
Hamiltonians sharing the same
generic tridiagonal-matrix structure
 \be
 H^{[M,N]}=\left[ \begin {array}{ccccccc}
 a_{-M}&b_{{-MN}}&0&\ldots&0&0&0
 \\{}c_{{-M+1}}&a_{-M+1}&
 b_{{-M+1}}&0&\ldots&0&0
 \\{}
 0&\ddots&\ddots &\ddots&\ddots&&0
 \\{}\vdots&\ddots&c_{{0}}&a_{0}&b_{{0}}&\ddots&\vdots
 \\{}0&&\ddots&\ddots&\ddots&\ddots&0
 \\{}0&0&\ldots&0&c_{{N-1}}&a_{N-1} &b_{{N-1}}
 \\{}0&0&0&\ldots&0&c_{{N}}&a_{N}
 \end {array}
 \right]\,
 \label{cotri}
 \ee
and sharing also the analogous characteristic
behavior in the limit of the large matrix dimensions,
 \be
 \lim_{M\to \infty}|a_{-M}|= \infty\,\ \ \ \ {\rm and}\ \ \ \
 \lim_{N\to \infty}|a_{N}|= \infty\,.
 \label{preas}
 \ee
In our constructions we
gave up some of the hidden
Bose-Hubbard Lie-algebraic symmetries (cf. \cite{Uwe})
but
we still managed to guarantee the existence
of the
EP degeneracies.
We pointed out that
in the language of physics, every EP limit
may be interpreted, in principle at least,
as
a genuine quantum
phase transition
\cite{passage} {\it alias\,}
a genuine quantum catastrophe \cite{catast}.
At each of the EP singularities, indeed,
the Hamiltonians
 $$
 H^{(BH)}({\rm i}\gamma^{(EP)},v^{(EP)},c^{(EP)})
 $$
as well as all of their generalized analogues
cease to be diagonalizable,
i.e., they cease to be acceptable
as representations of an observable.

In our present indirect continuation
of the latter studies
we decided to turn attention
from the existence
and influence
of such anomalies
to the conventional
quantum-mechanical dynamical regime in which
the system is controlled by
the complex Hamiltonian (\ref{cotri})
which is acceptable as an energy operator
in quantum mechanics,
being fully compatible with the theoretical postulates
as well as with the
conventional
probabilistic interpretation
of experiments.

\subsection{Amended factorization ansatz\label{faan}}

The date of birth
of quantum mechanics
can be identified with the publication
of the Heisenberg's paper \cite{WH}
which is going to be 100 years old this year.
Quantum observables
were interpreted there as eigenvalues of
an appropriate matrix.
In this spirit it makes sense to
study, first of all, various special forms of matrices
which represent, first of all, the
energy.

%
%

%

In a way mentioned in preceding section \ref{osinghovi},
a key technical role
can be then played by
Green's functions
represented in the continued-fraction form.
Surprisingly enough,
such a technique has not been used
in the case of the ``doubly infinite matrix'' models
as sampled by Eq.~(\ref{cotri})
\textcolor{black}{(cf. also, in this respect,
the
phenomenological motivation and
study
of the ``doubly infinite matrices'' in
Ref.~\cite{[9]}).
We have to admit that for us, precisely the existence of the
``doubly infinite'' matrix representation of
the non-Hermitian Bose-Hubbard model of Eq.~(\ref{Ham1})
was one of the key justifications of our choice of the model
as a ``suitable benchamark''}.
Let us now describe its consistent
\textcolor{black}{continued-fraction treatment and}
construction.
Filling the gap.

Under our present assumption
(\ref{preas})
an appropriate choice of the
separate factor matrices in Eq.~(\ref{finkit})
would be as follows,
%
%
 \be
 {\cal U}=
  \left[ \begin {array}{cccc|c|cccc}
  1&0&0&\ldots&0&\ldots&\ldots&0&0
  \\c_{{-M+1}}{f_{{-M}}}&1&0&\ddots&\vdots&&&\ldots&0
  \\{}0&\ddots&\ddots&\ddots&0&&&&\vdots
 \\{}\vdots&\ddots&c_{{-1}}{f_{{-2}}}&1&0&0&\ldots&\ldots&0\\
 \hline
 {}0&\ldots&0&c_{{0}}{f_{-1}}&1&b_{{0}}{f_1}&0&\ldots&0\\
 \hline
 {}0&\ldots&\ldots&0&0&1&b_{{1}}{f_{{2}}}&\ddots&\vdots
 \\{}\vdots&&&&0&\ddots&\ddots&\ddots&0
 \\{}0&\ldots&&&\vdots&\ddots&0&1&b_{{N-1}}{f_{{N}}}
 \\{}0&0&\ldots&\ldots&0&\ldots&0&0&1
 \end {array} \right]\,,
 \label{levy}
 \ee
%
%
 \be
 {\cal L}=
  \left[ \begin {array}{cccc|c|cccc}
  1&{f_{{-M}}}b_{{-M}}&0&\ldots&0&\ldots&\ldots&0&0
  \\0&1&{f_{{-M+1}}}b_{{-M+1}}&\ddots&\vdots&&&\ldots&0
  \\{}0&\ddots&\ddots&\ddots&0&&&&\vdots
 \\{}\vdots&\ddots&0&1&{f_{{-1}}}b_{{-1}}&0&\ldots&\ldots&0
 \\
 \hline
 {}0&\ldots&0&0&1&0&0&\ldots&0
 \\
 \hline
 {}0&\ldots&\ldots&0&{f_{{1}}}c_{{1}}&1&0&\ddots&\vdots
 \\{}\vdots&&&&0&\ddots&\ddots&\ddots&0
 \\{}0&\ldots&&&\vdots&\ddots&{f_{{N-1}}}c_{{N-1}}&1&0
 \\{}0&0&\ldots&\ldots&0&\ldots&0&{f_{{N}}}c_{{N}}&1
 \end {array} \right]\,.
 \label{pravy}
 \ee
Thus, the
bound-state
or resonant energies $E_n$
have to be sought
as roots of secular
equation
 \be
 \det ({H}-z)=0
 \,
 \label{secularum}
 \ee
i.e., under the overall regularity
assumptions as accepted in section \ref{osinghovi},
as poles of the Green's function.


Once we accept the shift of attention
to the quantum dynamical regime
which is far from EP singularities,
the behavior of the system
becomes better accessible to its experimental
non-quantum simulations, say, in the framework of
classical optics \cite{Christodoulides}.
Secular Eq.~(\ref{secularum})
yields then all of the real or complex
energy roots $z=E_n$
and, in general, it degenerates to the formally simpler rule
 \be
 \det {\cal F}(z)=0
 \,.
 \label{fsecularum}
 \ee
In full analogy with the
construction of section \ref{osinghovi},
the diagonal matrix ${\cal F}$
with elements
 $$
 1/f_{-M}, 1/f_{-M+1}, \ldots, 1/f_{N-1}, 1/f_N\,
 $$
has to be defined by recurrences
 \be
 f_{-j}=\frac{1}{a_{-j}-E-c_{-j}f_{-j-1}b_{-j-1}}\,,\ \ \ \
 j=M, M-1,\ldots,2 ,1
 \label{ucf}
 \ee
(with, formally, vanishing initial $f_{-M-1}=0$) and
 \be
 f_k=\frac{1}{a_k-E-b_kf_{k+1}c_{k+1}}\,,\ \ \ \
 k=N, N-1,\ldots,2 ,1
 \label{lcf}
 \ee
(where we use $f_{N+1}=0$).
We will also assume
that at a generic value of variable $z$ or $E$,
these recurrences remain well defined, i.e.,
that the denominators do not
vanish,
 \be
 1/f_{-M} \neq 0\,,
 \ \ \
 1/f_{-M+1} \neq 0\,,
 \ \ldots\  {\rm and} \ \ \ 1/f_{N} \neq 0\,,
 \ \ \ 1/f_{N-1} \neq 0\,,
 \ \ldots\,.
 \label{regu}
 \ee
Under this assumption, what differs
from the procedure of section~\ref{osinghovi}
is the necessity of the matching.

\begin{lemma}. \label{lemmaone}
In the regular, non-terminating cases,
the Green's function $f_0(z)$
of matrix Hamiltonians (\ref{cotri}) with any $M \leq \infty$ and $N \leq \infty$
and with an asymptotic growth property (\ref{preas})
can be defined,
at a suitable real or complex $E=z$,
in terms of a pair of
continued fractions of Eqs.~(\ref{ucf}) and (\ref{lcf}),
viz., by formula
 \be
 f_0(z)=\frac{1}{a_{0}-z-c_{0}f_{-1}(z)b_{-1}-b_0f_{1}(z)
 c_{1}}\,.
 \label{umacf}
 \ee
\end{lemma}
\begin{proof}.
Our regularity assumption (\ref{regu})
makes the asymptotic part of ${\cal F}$
invertible. At a generic parameter not in spectrum,
$z\notin \{ E_n\}$, we may
invert Schr\"{o}dinger operator,
 \be
 \frac{1}
 {H-E}=
 {\cal L}^{-1}\,
 {\cal F}^{-1}\,
 {\cal U}^{-1}\,.
 \label{[9h]}
 \ee
At any eigenvalue $z=E_n$
we have to satisfy secular Eq.~(\ref{fsecularum})
so that the regularity  (\ref{regu}) can only
be required to hold for $1/f_{j}$ at $j \neq 0$.
In other words,
we may
introduce a dedicated symbol for $1/f_{0} \equiv 1/G$
and, having
omitted all of the non-vanishing factors $1/f_{j}$
from Eq.~(\ref{fsecularum})
we arrive at the relation
 \be
 \lim_{z \to E_n}\frac{1}{G(z)}=0\,
 \label{maxs}
 \ee
i.e., at our ultimate
secular equation.
\end{proof}

\subsection{Wave functions}

In comparison with the methodically motivated
matrix factors (\ref{lowkit})
and factorized Schr\"{o}dinger
Eq.~(\ref{1finkit}) of section \ref{osinghovi},
we now have to deal with the more sophisticated
partitioned
matrix factors (\ref{levy}) and (\ref{pravy}).
They are still easily invertible yielding, in particular, formula
 $$
 \left[ \begin {array}{cccc|c|cccc}
 %
 \ddots& \ddots&\ddots&\vdots&\vdots&\vdots&&&
 \\{}\ldots&-c_{{-2}}f_{{-3}}&1&0&0&0&\ldots&&
 \\{}\ldots&c_{{-2}}f_{{-3}}c_{{-1}}f_{{-2}}&-c_{{-1}}
 f_{{-2}}&1&0&0&0&\ldots&
 \\
 \hline
 \ldots&-c_{{-2}}f_{{-3}}c_{{-1}}f_{{-2}}c_{{0}}f_{{-1}}
 &c_{{-1}}f_{{-2}}c_{{0}}f_{{-1}}&-c_{{0}}f_{{-1}}&1&-b_{{0}}f_{{1}}
 &b_{{0}}f_{{1}}b_{{1}}f_{{2}}&-b_{{0}}f_{{1}}b_{{1}}f_{{2}}b_{{2}}f_{{3}}&\ldots
 \\
 \hline
 &\ldots&0&0&0&1&-b_{{1}}f_{{2}}
 &b_{{1}}f_{{2}}b_{{2}}f_{{3}}&\ldots
 \\{}&&\ldots&0&0&0&1&-b_{{2}}f_{{3}}&\ldots
 \\{}&&&\vdots&\vdots&\vdots&\ddots&\ddots&\ddots
 \end {array} \right]\,.
 $$
%
%
for
 $
 {\cal U}^{-1}
 $ as well as an analogous partitioned representation of ${\cal L}^{-1}$.

The invertibility of both of these matrices is essential
because in a full conceptual analogy with section \ref{osinghovi},
the factorized Schr\"{o}dinger equation
(cf. Eq.~(\ref{1finkit}))
can again be reduced to
its simplified isospectral form.
A modification of the construction only occurs
due to the necessity of
the replacement of the more elementary secular Eq.~(\ref{generi})
by its present, slightly more complicated analogue (\ref{maxs}).

It is also necessary to replace
the single set of recurrences (\ref{jedenact})
for wave functions
by its partitioned counterpart.
Thus, recalling Eq.~(\ref{2finkit})
in its present partitioned modification
 \be
 {\cal F}\,
 {\cal L}\,\overrightarrow{\psi^{}}=0\,
 \label{m2finkit}
 \ee
we have to omit its redundant middle line (i.e., the zeroth line)
due to its equivalence to secular Eq.~(\ref{fsecularum}).

This makes the rest of the diagonal matrix factor
${\cal F}$ in (\ref{m2finkit})
invertible so that we can drop it in (\ref{m2finkit}).
The role of the normalization constant
is now attributed to $\psi_0$. Relation (\ref{m2finkit})
becomes equivalent to the two sets of recurrences, viz.,
 \be
 \left[ \begin {array}{ccccc}
 -c_{{1}}f_{{1}}&1&0&0&\ldots
 \\{}c_{{2}}f_{{2}}c_{{1}}f_{{1}}&-c_{{2}}f_{{2}}&1&0&\ldots
 \\{}-c_{{3}}f_{{3}}c_{{2}}f_{{2}}c_{{1}}f_{{1}}
 &c_{{3}}f_{{3}}c_{{2}}f_{{2}}&-c_{{3}}f_{{3}}&1&\ddots\\
 \vdots&\ddots&\ddots&\ddots&\ddots
 \\
 \end {array} \right]\,
 \left( \begin {array}{c}
 \psi_{0}\\
 \psi_{1}\\
 \psi_{2}\\
 \vdots\\
 \end {array} \right)=0\,
 \label{3bfinkit}
 \ee
and
 \be
 \left[ \begin {array}{cccc}
 \ddots&\ddots&\ddots&\vdots
 \\ \ddots&-f_{{-3}}b_{{-3}}&f_{{-3}}b_{{-3}}f_{{-2}}b_{{-2}}
  &-f_{{-3}}b_{{-3}}f_{{-2}}b_{{-2}}f_{{-1}}b_{{-1}}
 \\{}\ddots&1&-f_{{-2}}b_{{-2}}&f_{{-2}}b_{{-2}}f_{{-1}}b_{{-1}}
 \\{}\ldots&0&1&-f_{{-1}}b_{{-1}}
 \\
 \end {array} \right]\,
 \left( \begin {array}{c}
 \vdots\\
 \psi_{-2}\\
 \psi_{-1}\\
 \psi_{0}\\
 \end {array} \right)=0\,,
 \label{3afinkit}
 \ee
with the former ones running downwards,
and with the latter set of
recurrences running upwards.
The same (though potentially $E_n-$dependent)
initial value of $\psi^{}_0=\psi^{(n)}_0\neq 0$ is shared
by both of these sets.

Whenever needed, an entirely analogous
pair of recurrences
could also be written to define the left eigenvectors
of $H$.
Due to the
generic non-Hermiticity of Hamiltonian matrix
$H \neq H^\dagger$,
such an additional construction might make
perfect sense, indeed \cite{twoSEs}.


\section{Hermitization and quasi-Hamiltonians\label{ultim}}

In preceding sections \ref{osinghovi} and \ref{ufoun}
the non-degenerate complex eigenenergies $E_n$ of
a tridiagonal complex Hamiltonian $H \neq H^\dagger$
were
treated as poles of a Green's function.
Now we are going to turn attention to the fact that
numerically, the localization of such poles is not easy.

\subsection{Singular values as eigenvalues of an auxiliary operator}

In a way which we learned from the recent mathematically oriented
preprint \cite{PS1} it makes sense to reduce
the study of resonances
supported by non-Hermitian
$H$ by the study of non-negative eigenvalues
$\sigma_n^2$ of the products $H^\dagger\,H$.
In the context of quantum physics, quantities $\sigma_n$
are called
singular values
of $H$ \cite{SV}.

In this framework the authors of Ref.~\cite{PS1}
have noticed that the singular values $\sigma_n$ themselves
can be interpreted as eigenvalues
of an auxiliary self-adjoint operator
 \be
  \widetilde{\mathbb H}=
  \left (\begin{array}{c|c}
  0&H\\
  \hline
  H^\dagger&0
  \end{array}
  \right )\,
  \label{desce}
  \ee
where
%
 \be
 H=
  \left[ \begin {array}{ccccc}
     a_1&b_1&0
 &0&\ldots
   \\
   c_2&a_2&b_2&0&\ddots
   \\
 0
 &c_3&a_3&b_3&\ddots
   \\
 \vdots&\ddots&\ddots&\ddots&\ddots
    \\
 \end {array} \right]\ \
 \neq
 H^\dagger =
  \left[ \begin {array}{ccccc}
     a_1^*&c_2^*&0
 &0&\ldots
   \\
   b_1^*&a_2^*&c_3^*&0&\ddots
   \\
 0
 &b_2^*&a_3^*&c_4^*&\ddots
   \\
 \vdots&\ddots&\ddots&\ddots&\ddots
    \\
 \end {array} \right]\,.
 \label{trufinkit}
 \ee
In \cite{PS1}, incidentally,
Hamiltonians (\ref{trufinkit})
were required
complex-symmetric and
bounded, with property
\begin{equation}
\sup_{j}(|{a_j}|+|{b_j}|+|{c_{j+1}}|)<\infty\,.
\label{a1}
\end{equation}
In our present context,
via an appropriate re-indexing,
matrices (\ref{trufinkit})
can be, and have to be,
transformed
into their
Bose-Hubbard-like
versions $H=H^{[M,N]}$ of Eq.~(\ref{cotri}).
%

Needless to add that in our paper the operator-boundedness
constraint (\ref{a1})
is violated. This violation may concern either
the positive subscripts $j\gg 1$
(cf., typically, Eq.~(\ref{largeen}) in section \ref{osinghovi})
or both the positive
and
negative
large
subscripts (i.e., $|j|\gg 1$,
cf. Eq.~(\ref{preas})
and also Eq.~(\ref{reconfiteor}) in Appendix A below).
Our reasons of such a change of perspective
are phenomenological since our
Schr\"{o}dinger equations
have to describe non-degenerate sets of bound states
or resonances.
Marginally, it is possible to add that
one can work with non-Hermitian Hamiltonians
in both of these dynamical scenarios and regimes
\cite{Geyer}.

\subsection{Block-tridiagonalization transformation}

After the loss of the
reality of the spectrum
it makes sense to keep
the numerical tasks simple.
Thus, Pushnitzky
with \v{S}tampach \cite{PS1}
introduced the auxiliary
operator defined by
matrix (\ref{desce}).
They emphasized that such an operator
is self-adjoint and that it
yields the necessary non-negative singular values of $H$
simply as a subset of
non-negative
eigenvalues of $\widetilde{\mathbb H}$.

Although
these authors
only worked with a rather restricted
class of complex symmetric
matrices $H$,
their idea
works also in the present more general setting.
Having in mind the initial indexing (\ref{trufinkit}),
one simply introduces
a permutation of the basis
specified by the following partitioned auxiliary matrix
 \be
 \mathbb{V}=
  \left[ \begin {array}{cccccccc}
    1&0&0&0&\ldots&&&
   \\
    0&0&1&0&0&\ldots&&
   \\
    0&0&0&0&1&0&\ldots&
   \\
    \vdots&\ddots&&&\ddots&\ddots&\ddots&
   \\
   \hline
    0& 1&0&0&0&\ldots&&
   \\
    0&0&0&1&0&0&\ldots&
   \\
   0& 0&0&0&0&1&0&\ldots
   \\
   \vdots& \ddots&&&&\ddots&\ddots&\ddots
 \end {array} \right]\,
 \label{triceps}
  \ee
which is invertible. It is then easy to verify that
the permutation
given by the intertwining formula
 \be
 \widetilde{\mathbb H}\,\mathbb{V}=\mathbb{V}\,\mathbb{H}
 \ee
converts
the above-mentioned tilded
matrix $\widetilde{\mathbb H}$
into its
user-friendlier untilded
isospectral partner ${\mathbb H}$
which is block-tridiagonal,
 \be
 {\mathbb H}=
  \left[ \begin {array}{cccc}
     A_1&B_1&0
 &\ldots
   \\
   C_2&A_2&B_2&\ddots
   \\
 0
 &C_3&A_3&\ddots
   \\
 \vdots&\ddots&\ddots&\ddots
    \\
 \end {array} \right]\,.
 \label{matrix}
 \ee
At all subscripts $k$, moreover,
the separate two-by-two submatrices of ${\mathbb H}$, viz.,
 \be
 A_k=
\left (
\begin{array}{cc}
0&a_k\\
a_k^*&0
\ea
\right )\,,\ \ \ \ \
B_k=
\left (
\begin{array}{cc}
0&b_k\\
c_{k+1}^*&0
\ea
\right )\,,\ \ \ \ \
C_{k+1}=B_k^\dagger=
\left (
\begin{array}{cc}
0&c_{k+1}\\
b_k^*&0
\ea
\right )\,
 \label{[14]}
 \ee
are all sparse.

\subsection{Matrix continued fractions}

The block-tridiagonality of matrix (\ref{matrix})
finds its analogue in the ordinary tridiagonality of
Hamiltonian matrix $H$ in
Schr\"{o}dinger Eq.~(\ref{SEfinkit})
of section \ref{osinghovi}.
In this spirit the factorization
of the ordinary tridiagonal
Schr\"{o}dinger operator $H-E$ as given by Eq.~(\ref{finkit})
may find its immediate analogue
in relation
 \be
 \mathbb{H}-\sigma
 = {\mathbb U}\,
 {\mathbb F}\,
 {\mathbb L}\,\,
 \label{alfinkit}
 \ee
in which
the middle factor
${\mathbb F}$ would be a block-diagonal matrix
with two-by-two-matrix diagonal elements
 $
 1/F_j  $.
The other two large-matrix factors
 \be
 {\mathbb U}=
  \left[ \begin {array}{cccc}
  I&-U_2&0
 &\ldots
   \\
     0&I&-U_3&\ddots
   \\
 0
 &0&I&\ddots
   \\
 \vdots&\ddots&\ddots&\ddots
    \\
 \end {array} \right]\,,
 \ \ \ \ \ \
 {\mathbb L}=
  \left[ \begin {array}{cccc}
  I&0&0
 &\ldots
   \\
     -V_2&I&0
 &\ldots
   \\
   0&-V_3&I&\ddots
   \\
 \vdots&\ddots
 &\ddots&\ddots
   \\
 \end {array} \right]\,
 \label{blowkit}
 \ee
would be the block-bidiagonal
analogues of their predecessors of Eq.~(\ref{lowkit})
with the
properly
generalized two-by-two-matrix elements.

The parallels and consequences are obvious,
leading to the two-by-two matrix-continued-fraction (MCF, \cite{Nex}) recurrences
 \be
 F_k=\frac{1}{A_k-\sigma-B_k F_{k+1}C_{k+1}}\,,\ \ \ \
 k=N, N-1,\ldots,2 ,1\,
 \label{macf}
 \ee
where $F_{N+1}$ is the initial vanishing two-by-two matrix, and where,
in principle, $N \to \infty$.

Along these lines, the scalar scalar secular Eq.~(\ref{generi})
of section \ref{osinghovi}
becomes replaced by its two-by-two-matrix analogue
 \be
 \det F_1^{-1}(E_n)=0\,
 \ee
in a way which, occasionally, also found its numerical applications, i.a.,
in the context of anharmonic oscillators \cite{GG}.
\textcolor{black}{
Another
extension of the applicability of the present approach
using matrix continued fractions
could be also sought in the models of Ref.~\cite{[25b]}
for which some of the physical
Hamiltonian matrices appeared to be pentadiagonal.
Last but not least, a remark is to be made on the
slightly generalized
form of the fractions
needed for the study of certain higher-dimensional
systems.
In these cases (cf., e.g., \cite{multi}),
a natural limitation of applicability
can be seen to lie in a less obvious insight in the questions
of convergence
of the related generalized matrix continued fractions.}

\subsection{The case of Bose-Hubbard-like complex Hamiltonians $H$}

The Bose-Hubbard-like Hamiltonians $H^{[M,N]}$
of Eq.~(\ref{cotri})
are characterized by the two-sided asymptotic growth
(\ref{preas})
of their diagonal matrix elements.
Nevertheless, at any finite $M$ and $N$
they can still be read as just re-indexed forms of
finite non-Hermitian matrices of Eq.~(\ref{trufinkit}).
This means that once we turn attention to the two-sided asymptotic growth
(\ref{preas}) of their diagonals,
we can still replace Eq.~(\ref{matrix}) by its partitioned analogue
 \be
 {\mathbb H}=
  \left[ \begin {array}{ccc|c|ccc}
 \ddots&\ddots&\ddots&\vdots&&
    &\\
   \ddots&A_{-2}&B_{-2}&0&\vdots&
    &\\
   \ddots&C_{-1}&A_{-1}&B_{-1}&0&\ldots
    &\\
    \hline
  \ldots &0&C_0&A_0&B_0&0&\ldots
    \\
    \hline
    &\ldots
  &0& C_1&A_1&B_1&\ddots
   \\&& \vdots
 &0
 &C_2&A_2&\ddots
   \\&&
 &\vdots&\ddots&\ddots&\ddots
     \\
 \end {array} \right]\,.
 \label{bimatrix}
 \ee
This enables us to modify also factorization (\ref{alfinkit}) in terms of
the partitioned analogues of Eqs.~(\ref{levy}) and (\ref{pravy}). In this manner
one merely
replaces the lower-case
symbols (for the matrix elements of Hamiltonian $H$)
by the upper-case
symbols (representing the
two-by-two submatrices of
our present auxiliary matrix $\mathbb{H}$), i.e., one merely changes
$a_k \to A_k$, $b_k \to B_k$,
$c_{k+1} \to C_{k+1}$, etc.

The asymptotic- growth property (\ref{preas}) enables us
to consider the pair of two-by-two MCF recurrences
 \be
 F_{-j}=\frac{1}{A_{-j}-\sigma-C_{-j}F_{-j-1}B_{-j-1}}\,,\ \ \ \
 j=M, M-1,\ldots,2 ,1
 \label{biucf}
 \ee
(with, formally, vanishing initial $F_{-M-1}=0$) and
 \be
 F_k=\frac{1}{A_k-\sigma-B_kF_{k+1}C_{k+1}}\,,\ \ \ \
 k=N, N-1,\ldots,2 ,1
 \label{bilcf}
 \ee
(where we use $F_{N+1}=0$)
at any $M \leq \infty$ and $N \leq \infty$.
Ultimately we may
write down
also
the closed form of the Green's function
and we may expect that it remains well defined
in all of the generic non-terminating-continued-fraction cases.

\begin{lemma}. \label{lemmatwo}
The Green's function $G(z)$ associated with the Hermitized
Hamiltonian~(\ref{bimatrix})
can be defined by formula $G(z)=\det F_0(z)$ with
 \be
 F_0(z)=\left [{A_{0}-z-C_{0}F_{-1}(z)B_{-1}-B_0F_{1}(z)
 C_{1}}\right ]^{-1}\,,
 \label{umacf}
 \ee
i.e., in terms of the pair of the  two-by-two-matrix
continued fractions
defined by recurrences~(\ref{biucf}) and (\ref{bilcf}).
\end{lemma}
\begin{proof}.
The idea of the proof is the same as in Lemma \ref{lemmaone}.
After the
partitioning
of the operators and after the obvious replacements
$u_i \to U_i$,
$v_j \to V_j$ and $f_k \to F_k$
of the lower-case matrix elements by
the upper-case two-by-two matrices, the algebra and
structure of factors
remains unchanged.
In particular, the
standard continued fractions
defined by recurrences (\ref{ucf}) and (\ref{lcf})
become replaced by their MCF analogues.
\end{proof}

\textcolor{black}{Here, it is worth adding that
in practice,
a key to the applicability of all of the similar
matrix-continued-fraction results and formulae
will always be model-dependent.
In particular, interested readers may compare the very slow
matrix-continued-fraction convergence in paper \cite{GG}
with its extremely
quick rate in the case of the symmetrically
anharmonioc oscillator of paper \cite{SAO}.
Fortunately,
the present subclass of the specific
matrix continued fractions
is characterized by the
sparse structure of their
two-by-two submatrix coefficients
(cf. Eq.~(\ref{[14]})).
For this reason one may expect that
the convergence
may be expected quick here.
Competitive with
the other existing methods for computing singular values.
}

\textcolor{black}{Naturally, the latter statement
is valid
up to some rare-parameters extremes
in which the use of the matrix
continued fractions encounters its natural
imitations. One of such
end-of-applicability scenarios in which the matrix
continued fractions start to diverge
is given an explicit ``benchmark''
numerical illustration
in recent preprint \cite{MCF}.
In the language of physics
the onset of the matrix
continued fraction divergence usually means
a phase transition
to the dynamical regime characterized, typically, by
continuous spectrum \cite{thresh}.
}

\textcolor{black}{After all, an explicit confirmation
of the expectation convergence
can also be found provided in the most recent update
of preprint \cite{MCF}.
A generic illustrative example is studied there
algebraically as well as
numerically.
Incidentally, one of the serendipitious byproducts of the
corresponding more abstract algebraic analysis
appears important not only for the proof and speed of the convergence
but also for a demonstration of a robust
nature of these properties with respect to random perturbations:
The reason of validity of such an implication
for sensitivity analysis is that
the proof of convergence is based
on the truly robust fixed-point property
of the asymptotic forms of mappings (\ref{biucf}) or (\ref{bilcf}).
}

\section{Discussion\label{cotos}}

During the methodically motivated recollection of the
ordinary differential bound-state Schr\"{o}dinger equations
in Appendix A we had to distinguish between the quantum systems living
on half-axis ${\mathbb R}^+$ and on the full real line
${\mathbb R}$.
We showed that in the latter case one has to require
the growth of the potential (i.e.,
after discretization, of the main diagonal of Hamiltonian~(\ref{dvanact}))
at both of the coordinate infinities (cf. Eq.~(\ref{confiteor}) or,
for complex potentials, Eq.~(\ref{reconfiteor})).

In our present paper such an observation has been
transferred to the family of complex
tridiagonal-matrix
quantum bound-state Hamiltonians of Eqs.~(\ref{cotri}) + (\ref{preas})
which we called, in a way inspired by
Appendix B, Bose-Hubbard-like Hamiltonians representing the
general open systems with the complex (but still discrete) energy
spectra $\{E_n\}$.


For the purposes of illustration we recalled the manifestly
non-Hermitian Bose-Hubbard model of Ref.~\cite{Uwe}, and we
emphasized that for this model, the spectrum of the energies only
remains real (i.e., calculable by many standard and more or less
routine numerical procedures) inside a certain specific domain
${\cal D}^{(0)}$ of unitarity-compatible parameters.
\textcolor{black}{Naturally, such a choice could have been
complemented by the study of many other theoretical models which
could be followed by a more or less immediate confirmation by
experiments. Indeed, there exists a large number of physical systems
which attracted attention to the implications of non-Hermiticity in
physical applications. {\it Pars pro toto}, let us just start the
list by mentioning the thirty years old study \cite{[37]} of the
coupling of bound states to the open decay channels, the equally old
paper \cite{[44]} on the virtual-localized transitions, or just a
marginally younger paper \cite{[22]} dealing with the photonic
crystals. What followed a bit later were the descriptions of the
decays of atomic nuclei \cite{[33]}.}

\textcolor{black}{At the beginning of the new millennium we witnessed
a growth of interest in the anomalous forms of the propagation of
light \cite{[25]} and in the various new forms of the simulation of
the physical phenomena using non-Hermitian mathematics. Thus, the
family of the new examples of the real systems appeared to involve
some anomalous properties of the electronic circuits \cite{[23]} as
well as very real phenomenon of the reversal of the magnetic poles
of the Earth \cite{[38]}. Experiments have been performed using
optical microcavities \cite{[41]}. In parallel, multiple new
theoretical concepts emerged involving the vibrational surface modes
\cite{[42]} or the so called, Rabi splitting in semiconductors
\cite{[26]}. }

\textcolor{black}{Non-Hermiticity-related ideas also emerged in the
context of the study of chaos \cite{xxx,xxxa,xxxb}, especially via
its experimental simulations  in microwave billiards \cite{[28]}.
The other related analyses led to the constructions of very
innovative realistic pictures of the systems of quantum dots
\cite{[40]}, of the emergence of the so called Feshbach resonances
\cite{[46]}, etc.}

\textcolor{black}{ In fact, there is hardly any chance of making this
list complete. Thus, last but not least, let us close our hardly
exhaustive sample of the early stages of developments of the
non-Hermitian science by the last characteristic references to the
use of non-Hermiticity and continued fractions
in information processing \cite{[8]},
and to the description of the
phenomenon of the so called paramagnetic electron resonance
\cite{[32a]}.}

In our recent paper~\cite{zobecUwe} we revealed that
\textcolor{black}{in some of the models in the latter list, their
constructive analysis becomes not too easy, especially} when one
turns attention to the exterior of ${\cal D}^{(0)}$ or to the family
of the Bose-Hubbard-like models in which the spectra
(\textcolor{black}{i.e., in most cases,} the
energies of resonances) are complex. In the present continuation of
paper~\cite{zobecUwe} we managed to clarify the possibility of the
identification of these complex energies with the complex poles of a
properly generalized Green's function $G(z)$.

In this setting our present first key result is that
such a  generalization of the conventional Green's function
must be constructed as a function of a {\em pair\,}
of two independent analytic continued fractions.
We complemented this observation by the remark that
the conventional numerical localization
of the poles (which are, in general, complex)
may become almost prohibitively difficult
in practice.
We, therefore, recommended that
in such a situation, even the
search for an
incomplete information about the system
might make sense.

For the sake of definiteness we decided to
analyze just the possibility of
the search
for the mere singular values
$\sigma_n$ of $H$ \cite{SV}.
For this purpose,
we had to use
the method
based on the Hermitization $H \to {\mathbb H}$
as
developed by
Pushnitski and \v{S}tampach in \cite{PS1}.

The first step in this direction has been made
in our recent preprint \cite{MCF}.
We considered there just the case of the ordinary tridiagonal-matrix Hamiltonians
of Eq.~(\ref{SEfinkit}) (i.e., formally, with the choice of a fixed integer $M = -1$
in our $H^{[M,N]}$ of Eq.~(\ref{cotri})).
The key technical ingredient was that
once we were interested in the most common quantum models
with discrete spectra,.
we decided to
replace the bounded-operator assumption of paper \cite{PS1}
by the one-sided-growth requirement
 \be
 \lim_{N\to \infty}|a_{N}|= \infty\,
 \label{kupreas}
 \ee
which still enabled us to construct the Green's function
in terms of a single continued fraction.
Incidentally, this was a constraint which enabled us to
deliver, in particular, the proof of the MCF convergence.

In our present continuation of these efforts, our attention
has been extended
to the ``doubly-infinite-matrix'' Bose-Hubbard-like systems
characterized by the
replacement of Eq.~(\ref{kupreas}) by the
doublet of asymptotical-behavior constraints (\ref{preas}).
Among all of the eligible versions
of such quantum systems and models, the one with equal cut-offs
$M=N$ might be given a particular attention because its
matrix structure
would lead us back to the
systems with
${\cal PT}-$symmetry \cite{book}
and, in particular, to the
specific model of Graefe et al \cite{Uwe}.
Incidentally, we only paid attention here to the latter systems
in the case of the even number of
$2N$ bosons, with
their states described
by
$(2N+1)-$dimensional matrix Hamiltonian
of the form of Eq.~(\ref{cotri}).
Now, it is time to add that
it would be straightforward
to readapt the formulae to cover also the
systems with the odd numbers of bosons.
This task is easy and can be left to the readers as an exercise.

There exist
parallels between
the conventional continued-fraction representation
of Green's functions as sampled by Eq.~(\ref{Virendra})
and their present two-continued-fraction representation as sampled by
Eq.~(\ref{umacf}).
The analogies can be illustrated and traced back to the
single-particle models
of Appendix A.
Indeed, Schr\"{o}dinger Eqs.~(\ref{SEloc})
or (\ref{SEdis})
living on the whole real line
(with $\psi_n(x)\in L^2(\mathbb{R})$
in (\ref{SEloc}) or with
$\psi_n(x_k)\in \ell^2(\mathbb{Z})$ in (\ref{SEdis}))
can be compared with the models living
on the half-axis
(with $\psi_n(x)\in L^2(\mathbb{R}^+)$
in (\ref{SEloc}) or with
$\psi_n(x_k)\in \ell^2(\mathbb{N})$ in (\ref{SEdis})).

One comes to the conclusion that
in the former, full-line case, the confining
potential
must be asymptotically unbounded
at both of the ends of the real line
(cf. Eq.~(\ref{confiteor})).
The situation is different
in the half-axis scenario where
one usually deals just with
the single asymptotic condition of
growth of $V(x)$ at positive
$x\gg 1$ or $x_k\gg 1$.
This means that in
the
discrete analogue (\ref{cotri}) + (\ref{preas}) of the
differential
full-line bound-state problem (\ref{SEloc}) + (\ref{confiteor}),
our present amended Green's function $G(z)$
entering secular Eq.~(\ref{maxs})
can still be expected convergent because it is
defined in terms of the {\em two independent} continued fractions.

Naturally, an entirely analogous conclusion can be also drawn
in the context
of the singular values and matrix continued fractions
(cf., in particular, Eq.~(\ref{umacf})).

\section{Summary\label{summary}}

The
recent growth of interest in
the use of non-Hermitian
operators in quantum mechanics
opened also new horizons in the
study of quantum phase transitions.
With a persuasive illustration
provided by the Bose-Hubbard
Hamiltonian of Graefe et al \cite{Uwe}
it has been clarified that
the process of a typical
phase transition finds an appropriate mathematical
representation in the concept of exceptional point \cite{Heiss}.

The kinematics as well as
dynamics of this process remained to be a challenge.
Even using certain exactly solvable models (cf., e.g., \cite{passage},
with further references)
the analysis and description of quantum dynamics
in a vicinity of an exceptional point
remained far from straightforward,
requiring the formalism of non-Hermitian
quantum mechanics in its interaction-picture
representation \cite{NIP}.

In the light
of the lack of efficient description methods
one of them
has been proposed in our present paper.
We restricted our attention
to the exceptional-point-supporting class of phenomenological Hamiltonians
available in the next-to-elementary tridiagonal-matrix representation.
We emphasized that in the framework of conventional theory
one of the efficient construction techniques
is known to be the representation of Green's functions
using analytic continued fractions.

We proposed a generalization of the approach.
We claimed that once the Hamiltonian is admitted non-Hermitian, and once
its spectrum is admitted complex, one can still
construct the Green's functions using just an appropriate
upgrade of the recurrent continued-fraction-expansion
philosophy and mathematics.

Several related open questions have been solved.
Firstly, we managed to explain an important
formal difference between the standard
matrix Hamiltonians (as sampled, in our paper, by
the reference to the paper by Singh et al \cite{Singh})
and the non-standard
matrix Hamiltonians (as sampled by
the generalized Bose-Hubbard model of paper \cite{zobecUwe}).

We revealed that such a difference finds its
comparatively elementary explanation in a comparison of
the ordinary differential Schr\"{o}dinger equations
living on the half-axis of
coordinates $r \in (0,\infty)$ with the ones
living on the
full real line.
The one-parametric growth $N \to \infty$ of the dimension
of a conventional matrix $H$ (say, of Eq.~(\ref{SEfinkit}))
has to be replaced by the double limiting transition
with $M \to \infty$ and  $N \to \infty$ as sampled, e.g.,
in Eq.~(\ref{cotri}).

Our efforts were motivated by the increase of
difficulties accompanying the recent turn of interest to
the
quantum systems with non-Hermitian Hamiltonians.
Indeed, the localization of
the bound-state energies treated as the
poles of the continued-fraction Green's functions
only remained comparatively easy
under a guarantee that these values
must be real,
i.e.,
under a guarantee that the
Hamiltonian $H$ in question
remains Hermitian, or, at worst, quasi-Hermitian \cite{Geyer}.

In technical terms, we felt inspired by paper \cite{PS1}
in which the authors
recommended
to avoid the search for complex energies
and to reduce the problem to the mere computation of the
singular values $\sigma_n$.
Using this guidance we
also reinterpreted the singular values $\sigma_n$
as eigenvalues of a Hermitized Hamiltonian
{\it alias\,} quasi-Hamiltonian $\widetilde{\mathbb{H}}$
or, equivalently, of its
untilded isospectal avatar ${\mathbb{H}}$
possessing a block-tridiagonal
matrix structure.

In the natural ultimate step we pointed out
that the block-tridiagonality of ${\mathbb{H}}$ enables one to
construct also the related Green's function. We
only have to emphasize that
such a function has
a sophisticated form
constructed in terms of a
doublet of matrix continued fractions.

%
%
%
%

\newpage

\section*{Appendix A. Discrete version of Schr\"{o}dinger equation
as an elementary tridiagonal-matrix example}

One of the
best known
examples of
a tridiagonal-matrix (\ref{cotri})
emerges, in a purely numerical context, when
one tries to construct bound states using Schr\"{o}dinger equation
 \be
 -\frac{d^2}{dx^2}\,\psi_n(x)+
  V(x)\,\psi(x) =E_n\,\psi_n(x)\,,
 \ \ \ \ \ \ \
 \psi_n(x)\in L^2(\mathbb{R})\,
 \label{SEloc}
 \ee
with a real,
local and asymptotically growing confining potential
such that
 \be
 \lim_{x \to - \infty} V(x) = +\infty\,,\ \ \ \ \
 \lim_{x \to + \infty} V(x) = +\infty\,.
 \label{confiteor}
 \ee
The ordinary
differential
Eq.~(\ref{SEloc}) is
replaced by its discrete version
living on a finite equidistant lattice $\{x_k\}$,
 \be
 -\frac{\psi_n(x_{k-1})-2\,\psi_n(x_k)+\psi_n(x_{k+1})}{h^2}+V(x_k)\,
 \psi_n(x_k)
 =E_n\,\psi_n(x_k)\, ,
 \ \ \ \ \ \ \
 \psi(x_k)\in \ell^2(\mathbb{N})\,.
 \label{SEdis}
 \ee
After we subtract $2\,\psi_n(x_k)/h^2$
from both sides of the equation,
and after we absorb such a shift in a rescaled energy,
our
Hamiltonian
acquires the elementary tridiagonal-matrix form
 \be
 {H}=
 \left (
 \begin{array}{ccccc}
 V(x_{-M})&-1/h^2&&&\\
 -1/h^2&V(x_{-M+1})&-1/h^2&&\\
  &-1/h^2& \ddots & \ddots&\\
  & &\ddots&V(x_{N-1})&-1/h^2 \\
 &&&-1/h^2&V(x_{N})
 \ea
 \right )\,
 \label{dvanact}
 \ee
characterized by the special case
(\ref{confiteor}) of the above-mentioned
asymptotic-growth property
of the diagonal matrix elements~(\ref{preas}).

The off-diagonal
constants $b_j=c_{j+1}=-1/h^2$
reflect the equidistant-lattice nature of
kinetic energy.
For the other,
non-equidistant-lattice forms
of the sets of coordinates,
one only has to use a suitable modification
of the latter constants (see, e.g., \cite{hussin}).
Alternatively,
non-tridiagonal-matrix forms
of the kinetic energy
do also exist and
can be found described and used, e.g.,
in \cite{exact}.

\begin{figure}[h]                    
\begin{center}                         
\epsfig{file=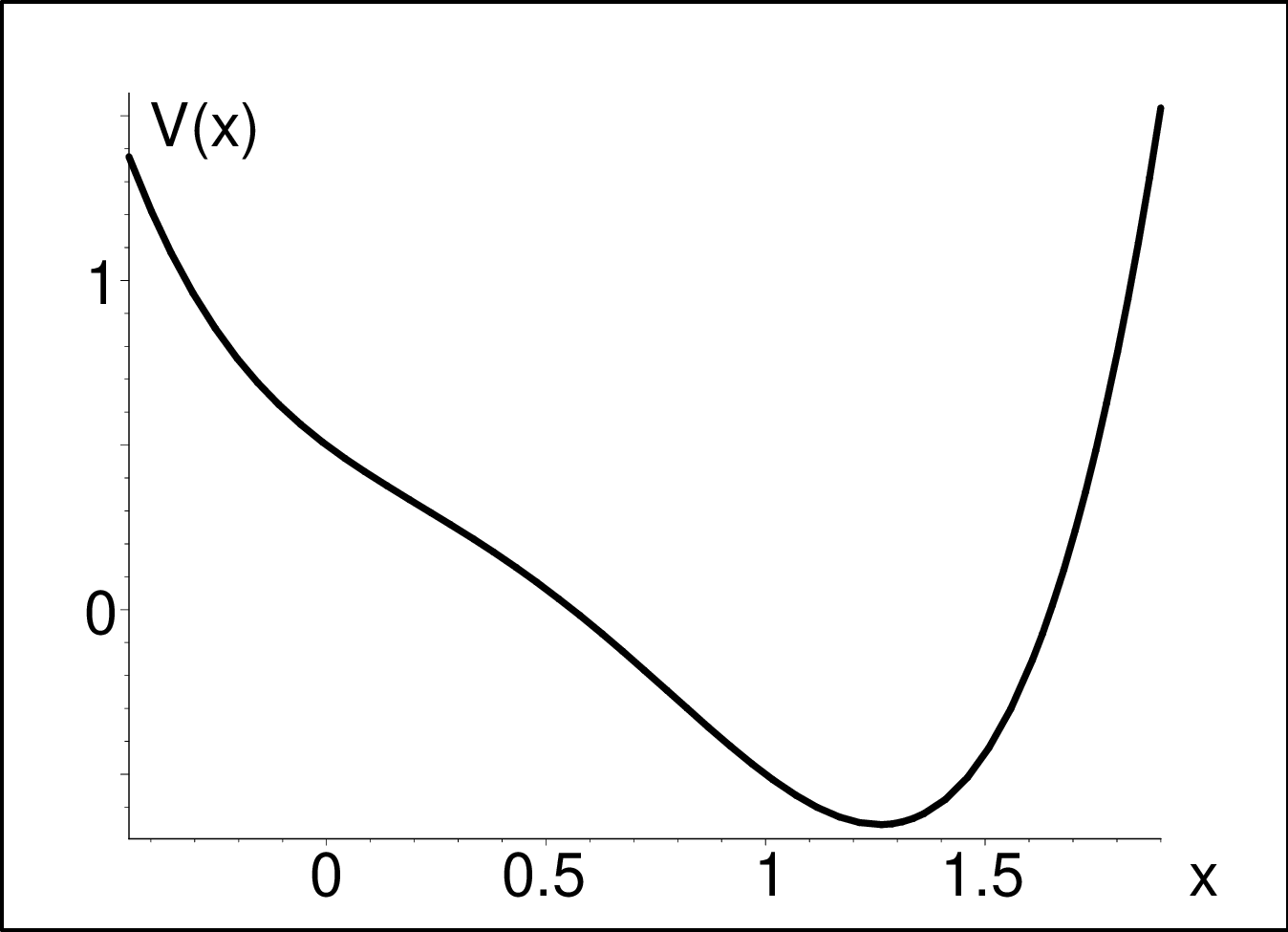,angle=270,width=0.35\textwidth}
\end{center}    
\caption{The shape of  potential V(x)
of Eq.~(\ref{Vbus}) (we choose $j=g=1$ in
Theorem Nr. 5 of Ref. \cite{BG}).
 \label{globe}}
\end{figure}

\begin{figure}[h]                    
\begin{center}                         
\epsfig{file=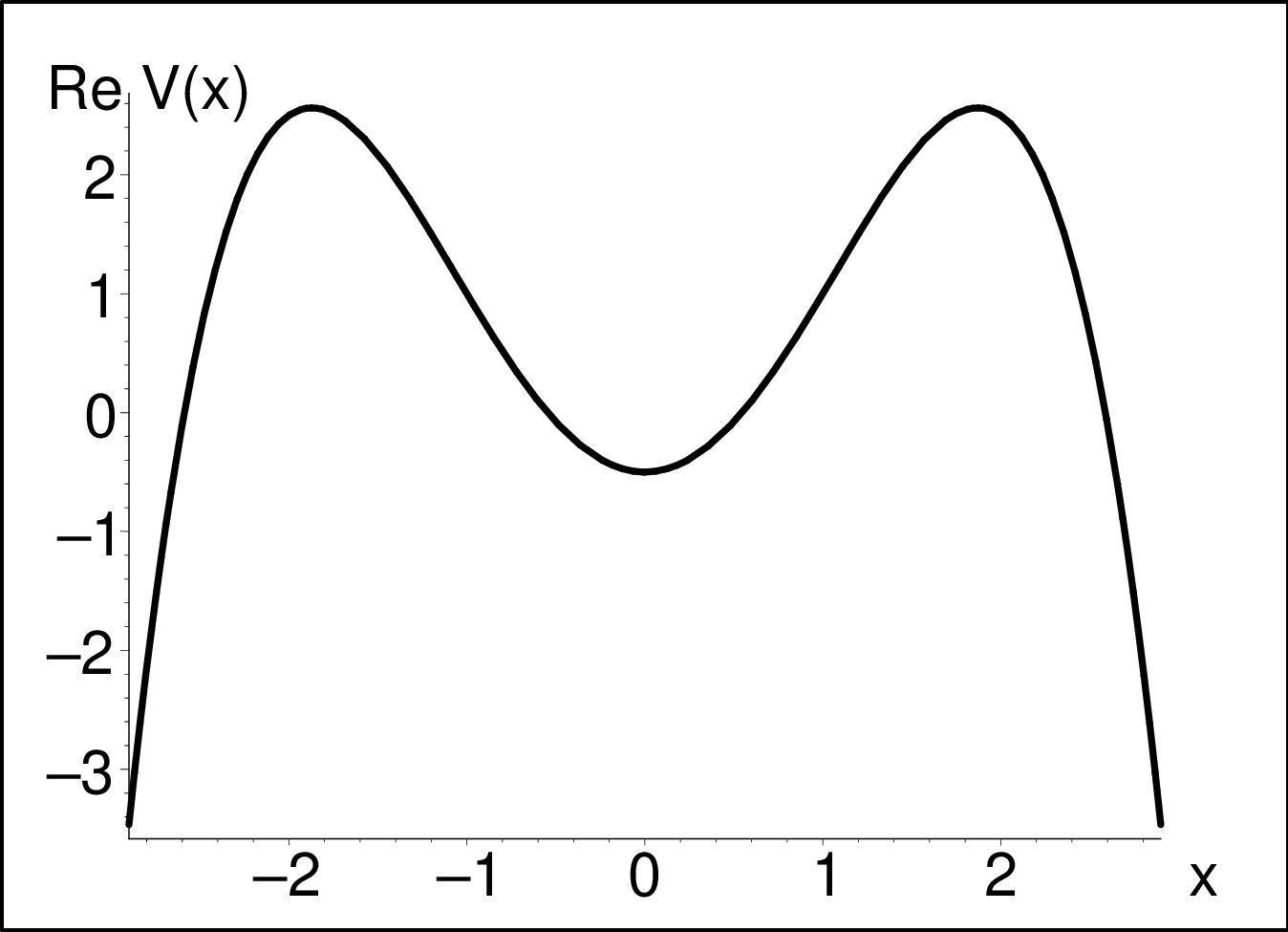,angle=270,width=0.35\textwidth}
\end{center}    
\caption{The shape of the real part of potential W(x)
of Eq.~(\ref{Wbus}) at $\eta=1$.
 \label{rglobe}}
\end{figure}

\begin{figure}[h]                    
\begin{center}                         
\epsfig{file=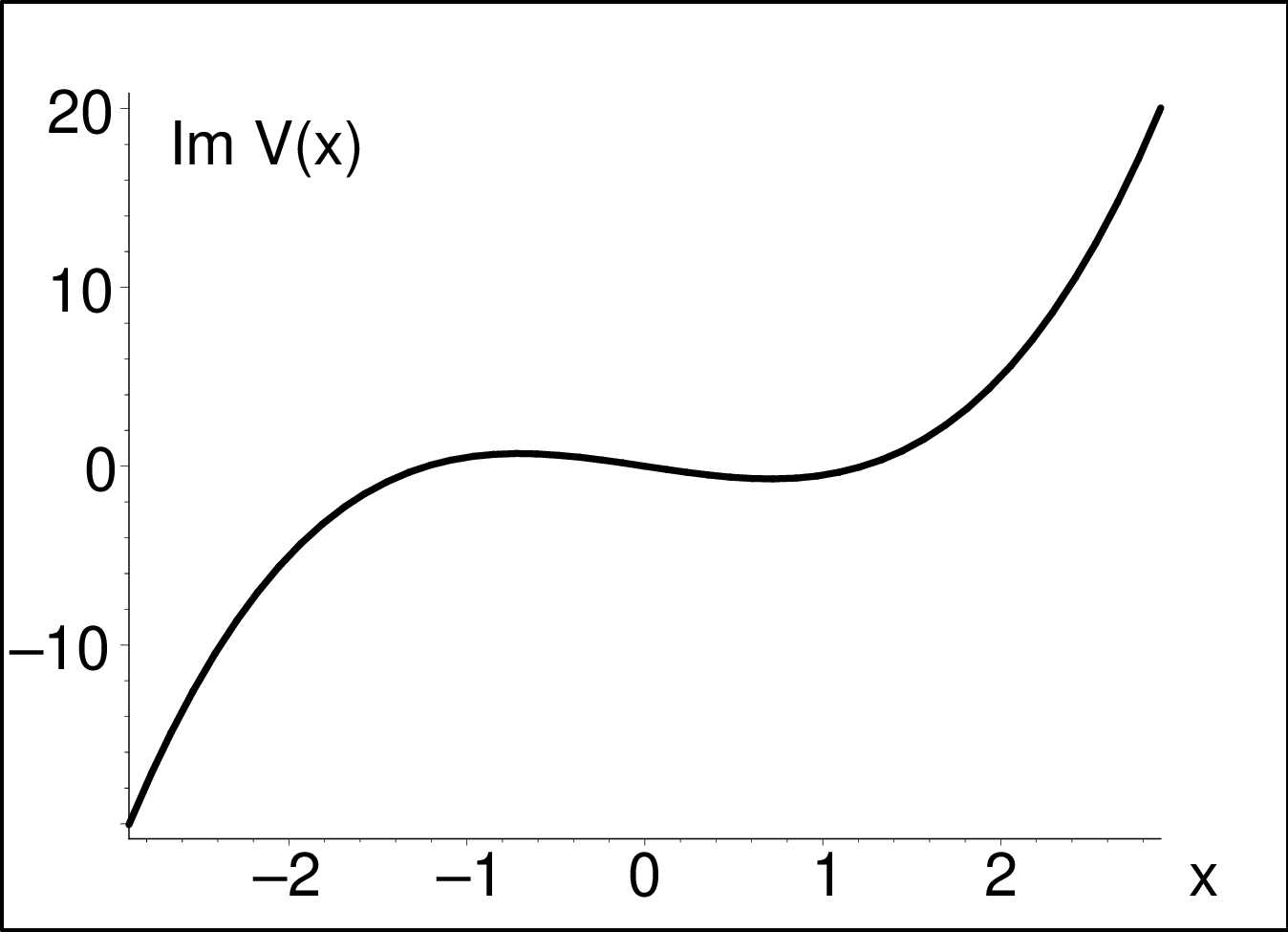,angle=270,width=0.35\textwidth}
\end{center}    
\caption{The shape of the imaginary part of potential W(x)
of Eq.~(\ref{Wbus}) at $\eta=1$.
 \label{iglobe}}
\end{figure}

A return
to the continuous case could be
achieved via limiting transition of $M \to \infty$, $N \to \infty$,
and $h \to 0$.
The dynamics then remains
controlled by the diagonal matrix elements
$a_j = V(x_j)$ so that
under the conventional assumption of their reality,
conditions of Eq.~(\ref{confiteor})
can be perceived
as a special case of the above-mentioned
matrix asymptotic growth
(\ref{preas}).

In comparison, constraint
(\ref{preas})
involving the absolute values
of the diagonal matrix elements of $H$
was introduced as a purely mathematical
guarantee of convergence of
the continued-fraction recurrences (\ref{ucf}) and (\ref{lcf}).
Nevertheless, once we also modify
conditions
(\ref{confiteor}) to read
 \be
 \lim_{x \to - \infty} |V(x)| = +\infty\,,\ \ \ \ \
 \lim_{x \to + \infty} |V(x)| = +\infty\,
 \label{reconfiteor}
 \ee
we may also admit
complex
potentials $V(x) \in \mathbb{C}$
in Eqs.~(\ref{SEloc}) and (\ref{SEdis}) (cf. \cite{book}).

For an explicit illustration of such an extension of the conventional
postulates of
quantum theory
let us recall
the toy model of paper \cite{BG}.
Buslaev with Grecchi proved there that the entirely conventional
Schr\"{o}dinger Eq.~(\ref{SEloc})
with a real and asymptotically growing potential
 \be
 V^{(BG)}(x)={x}^{2} \left( x-1 \right) ^{2}-x+1/2
 \label{Vbus}
 \ee
yielding the real and discrete bound state spectrum
(see the shape of this potential in Figure \ref{globe})
is isospectral with
all of its $\eta-$numbered non-Hermitian
alternatives with
 \be
V^{(\eta)}(x)=-\frac{1}{4}\,(x-i\,\eta)^4+\frac{1}{4}
(x-i\,\eta)^2\,,\ \ \ \ \ \eta>0\,.
 \label{Wbus}
 \ee
Asymptotically, the absolute values
of the latter potentials grow, $|V^{(\eta)}(x)| = x^4/4 + \ldots$,
but they are dominated by the real part which
is asymptotically decreasing to minus infinity
(see its shape in Figure \ref{rglobe}).
An analogous, albeit one-sided,
anomalous asymptotic behavior characterizes also
the asymptotically subdominant
imaginary part of the potential
(cf. Figure \ref{iglobe}).

A parallel now emerges between the non-Hermitian discrete
Schr\"{o}dinger equations with interactions $V^{(\eta)}(x)$
and the
Bose-Hubbard matrix models with
Hamiltonians $H^{(BH)}({\rm i}\gamma,v,c)$.
Indeed, once we fix $v=1$ and
choose $c=0$, we reveal that
the two $|x| \to \infty$ asymptotics of Figure \ref{iglobe}
find their
immediate respective
analogues in the $M \gg 1$ and $N \gg 1$
asymptotics
of
the purely imaginary diagonal matrix elements of matrices
$H^{(N)}_{(CBH)}(\gamma)$ or
$H^{(N)}_{(non-BH)}(\gamma)$
of Appendix~ B.

\newpage

\section*{Appendix B. Non-Hermitian Bose-Hubbard model as a methodical guide}

During the study of the
analytically continued version $H^{(BH)}({\rm i}\gamma,v,c)$
of the conventional
Bose-Hubbard Hamiltonian (\ref{Ham1}),
the authors of paper \cite{Uwe} managed to simplify several
less essential technicalities.
One of the methodically most important simplifications
resulted from the choice of a
vanishing strength $c=0$ of interaction
between particles.

In the framework of
the resulting, still fairly realistic second-quantized model
(offering, i.a., a quantitative insight in the
mechanism of the so called Bose-Einstein
bosonic condensation)
the authors also fixed the control of tunneling
and set $v=1$.
They
pointed out that in the system
the number of particles ${\cal N}$
is conserved so that
its choice reduces the Hamiltonian
to a comparatively elementary
$K$ by $K$ complex matrix $H^{(K)}_{(CBH)}(\gamma)$
of the tridiagonal-matrix form (\ref{cotri})
with  $K={\cal N}+1$
and property (\ref{preas})
(see formula Nr. 8 in {\it loc. cit.}).

The authors made an ample use of
the exact solvability of the
model with $c=0$. Indeed, at $K=2$
they had to deal then with the most elementary Hamiltonian matrix
\be
 H^{(2)}_{(CBH)}(\gamma)=
 \left[ \begin {array}{cc} -i{\it \gamma}&1
 \\\noalign{\medskip}1&i{\it
 \gamma}
 \end {array} \right]\,
 \label{dopp2}
 \ee
generating the
pair of
bound states with energies $E_\pm = \pm \sqrt{1-\gamma^2}$ and with
the two eligible non-Hermitian exceptional-point (EP)
degeneracies at $\gamma^{(EP)}_\pm=\pm 1$. Similarly, at $K=3$
one has
 \be
H^{(3)}_{(CBH)}(\gamma)=\left[ \begin {array}{ccc} -2\,i\gamma&
\sqrt{2}&0\\\noalign{\medskip}\sqrt{2}&0&
\sqrt{2}\\\noalign{\medskip}0&\sqrt{2}&2\,i\gamma\end {array}
\right]\,
  \label{3wg}
 \ee
and one obtains
$E_0=0$ and $E_\pm = \pm 2\,\sqrt{1-\gamma^2}$ while still
$\gamma^{(EP)}_\pm =\pm 1$, etc.
%

Marginally, it is worth adding that
the first nontrivial ``single-particle'' item  (\ref{dopp2})
of the $K-$numbered sequence
of the simplified Hamiltonian matrices
is in fact the most popular non-Hermitian matrix
used, for illustration purposes,
in many recent papers on
the so called pseudo-Hermitian \cite{ali} {\it alias\,}
${\cal PT}-$symmetric \cite{making} reformulations of
quantum mechanics.

For our present methodical and illustrative purposes we
decided to restrict attention,
for the sake of transparency of the notation, to the
subset of the latter family of toy models comprising the
even numbers of particles, with the next, four-particle item
 \be
 H^{(5)}_{(CBH)}(\gamma)=\left[
 \begin {array}{ccccc} -4\,i\gamma&2&0&0&0
 \\
 \noalign{\medskip}2&-2\,i\gamma&\sqrt {6}&0&0
 \\
 \noalign{\medskip}0&\sqrt{6}&0&\sqrt {6}&0\\\noalign{\medskip}0&0&
 \sqrt {6}&2\,i\gamma&2\\\noalign{\medskip}0&0&0 &2&4 \,i\gamma\end
 {array} \right]\,
 \label{petpa}
 \ee
etc. This means that we will work just with
$M=N \leq \infty$ and $K={\cal N}+1 = 2N+1$
in our present general Hamiltonian-matrix
ansatz (\ref{cotri}) + (\ref{preas}).

In the context of physics,
the authors of paper \cite{Uwe}
saw
a key phenomenological
appeal of models $H^{(K)}_{(CBH)}(\gamma)$
in the
existence of the truly remarkable
exceptional points $\gamma^{EP}$ of
the $K-$th (i.e., of the maximal) order.
In this spirit, these models may be found generalized, e.g.,
in \cite{mujBH} where, typically, the EP-supporting complex-symmetric
$K=5$ Hamiltonian matrix~(\ref{petpa})
found its complex-symmetric Bose-Hubbard-like alternative
 \be
 H^{(5)}_{(non-BH)}(\gamma)=\left[
 \begin {array}{ccccc} -4\,i\gamma&8&0&0&0
 \\
 \noalign{\medskip}8&-2\,i\gamma&i\sqrt {54}&0&0
 \\
 \noalign{\medskip}0&i\sqrt{54}&0&i\sqrt {54}&0\\\noalign{\medskip}0&0&
 i\sqrt {54}&2\,i\gamma&8\\\noalign{\medskip}0&0&0 &8&4 \,i\gamma\end
 {array} \right]\,
 \label{zpetpa}
 \ee
etc.

The alternative models
lost most of the original Bose-Hubbard symmetries
after such a generalization.
One had to diagonalize them numerically \cite{BH18}.
Moreover, even near the EP singularities,
the EP-unfolding patterns appeared to
resist a perturbation-approximation treatment.
Indirectly, the emergence of all of these
difficulties  also motivated our present
continued-fraction (i.e., semi-analytic) approach and
constructions of the resolvents and of the Green's functions.

Naturally, our much more immediate motivation originated in physics.
The source of appeal
of the study of the behavior of the
open quantum systems near their EP
dynamical extremes can be seen
in its deeply innovative role in
the
study of
the theoretical aspects of the
quantum phase transitions (cf. also \cite{passage})
as well as
in
its possible experimental relevance.
Both of these aspects of the
process of unfolding of the EP degeneracies
has been also properly emphasized in the dedicated papers like \cite{Uwe} or
\cite{mujBH},
etc.

\newpage


\begin{thebibliography}{999}
\bibitem{Gar}
Incidentally, these atytempts were
inspired and
motivated by
one of the anonymous referees
of paper \cite{MCF}.
\end{thebibliography}
\end{document}